\documentclass[%
aip,%
jmp,%
 amsmath,amssymb,
 onecolumn,%
author-numerical,%
]{revtex4-1}
\usepackage{amsmath, amsthm, amscd, amsfonts, amssymb, graphicx, color, enumitem, dsfont}
\usepackage[bookmarksnumbered, colorlinks, plainpages]{hyperref}
\hypersetup{colorlinks=true,linkcolor=red, anchorcolor=green, citecolor=cyan, urlcolor=red, filecolor=magenta, pdftoolbar=true}

\newtheorem{theorem}{Theorem}[section]
\newtheorem{lemma}[theorem]{Lemma}
\newtheorem{proposition}[theorem]{Proposition}
\newtheorem{corollary}[theorem]{Corollary}
\theoremstyle{definition}

\newtheorem{example}{Example}

\newtheorem*{remark}{Remark}

\begin{document}

\setcounter{page}{1}

\title{Iterates of quantum operations}

\author{J\'ozsef Zsolt Bern\'ad}
\affiliation{Peter Gr\"unberg Institute (PGI-8), Forschungszentrum J\"ulich, D-52425 J\"ulich, Germany}
\email{j.bernad@fz-juelich.de}
\affiliation{
Institut f\"{u}r Angewandte Physik, Technische Universit\"{a}t Darmstadt, D-64289 Darmstadt, Germany
}




\begin{abstract}
Iterates of quantum operations and their convergence are investigated in the context of mean ergodic theory. We discuss in detail the convergence of 
the iterates and show that the uniform ergodic theorem plays an essential role. Our results will follow from some general 
theorems concerning completely positive maps, mean ergodic operators, and operator algebras on Hilbert spaces. A few examples of both finite and infinite dimensional Hilbert spaces are presented as well.
\end{abstract} \maketitle

\section{Introduction and preliminaries}
\label{I}

Quantum system are never perfectly closed and therefore interactions always take place with certain parts of their environment. This is the concept of open quantum systems, where
we wish to neglect up to some extent the dynamics of the environment but follow the state changes of the central system. These processes of state changes are called quantum operations. Quantum operations 
in the Hilbert space formulation of 
quantum mechanics are obtained in the following way \cite{Hellwig}: we start with an uncorrelated joint state of the central system and the environment; it is followed by a joint unitary evolution; and then
an observer measures a property of the environment, described by projective operations acting only on the environment. It has been shown by Kraus \cite{Kraus} that quantum operations are 
completely positive. They play an important role in theory and applications of uniformly continuous completely positive dynamical semigroups \cite{GKS, Lindblad, Spohn}. Physical applications are present in
various subfields of quantum information processing, like quantum computing \cite{Nielsen} or quantum control theory \cite{Wiseman}, but also in problems related to the foundations of quantum mechanics \cite{Joos}.

Throughout this manuscript, we consider a separable Hilbert space $\mathcal{H}$ \cite{comment} with inner product $ \langle.\,,.\rangle$ which
is conjugate linear in the first and linear in the second variable. The norm $\langle x,x \rangle^{1/2}$ of any element $x$ in $\mathcal{H}$ will be denoted
by $\|x\|$. The adjoint of a bounded linear operator $A$ on $\mathcal{H}$ is the unique
operator $A^\dagger$ satisfying $\langle A^\dagger x, y \rangle=\langle x, Ay \rangle$ for all $x,y$ in $\mathcal{H}$. The set of all bounded linear operators on $\mathcal{H}$ is denoted by 
$\mathcal{B}(\mathcal{H})$, which is a Banach space with respect to the operator norm
\begin{equation}
 \|A\|=\sup\{\|Ax\|: x \in \mathcal{H}, \|x\| \leqslant 1\}. \nonumber
\end{equation}
The Banach space of trace class operators is defined as
\begin{equation}
   \mathcal{B}_1(\mathcal{H})=\left \{ X \in \mathcal{B}(\mathcal{H}): \|X\|_1=\mathrm{Tr}\left\{\sqrt{ X^\dagger X}\right\} <\infty \right \}, \nonumber
\end{equation}
in which the set of quantum states or density operators is given by
\begin{equation}
 \mathcal{D}(\mathcal{H})=\left \{\rho \in \mathcal{B}_1(\mathcal{H}): \rho \geqslant 0, \quad \mathrm{Tr} \left\{\rho \right\}=1 \right \}. \nonumber
\end{equation}
Let us consider the $n$-dimensional Hilbert space $\mathbb{C}^n$ and a linear map $\phi: \mathcal{B}(\mathcal{H}) \rightarrow \mathcal{B}(\mathcal{H})$. Operators in $\mathcal{B}(\mathcal{H} \otimes \mathbb{C}^n)$
are $n \times n$ matrices with entries $A_{ij}$ acting on the Hilbert space $\mathcal{H}$ for each $1\leqslant i,j\leqslant n$. We define the linear map
$\phi^{(n)}:  \mathcal{B}(\mathcal{H} \otimes \mathbb{C}^n) \rightarrow \mathcal{B}(\mathcal{H} \otimes \mathbb{C}^n)$ by
\begin{equation}
 \phi^{(n)}\begin{pmatrix}
A_{11} & \cdots & A_{1n} \\
\vdots & \ddots & \vdots \\
A_{n1} & \cdots & A_{nn}
\end{pmatrix} \mapsto \begin{pmatrix}
\phi(A_{11}) & \cdots & \phi(A_{1n}) \\
\vdots & \ddots & \vdots \\
\phi(A_{n1}) & \cdots & \phi(A_{nn})
\end{pmatrix}. \nonumber
\end{equation}
If $\phi^{(n)}$ is positive for all $n\geqslant 1$ then $\phi$ is a completely positive map. A linear map $\phi: \mathcal{D}(\mathcal{H}) \rightarrow \mathcal{D}(\mathcal{H})$ of the form
\begin{equation}
 \phi(\rho)=\sum_i V_i \rho V^\dagger_i, \label{phi} 
\end{equation}
with a family of operators $V_i \in \mathcal{B}(\mathcal{H})$ such that \cite{Kraus}
\begin{equation}
 \left \langle x, \sum_i V^\dagger_i V_i x \right \rangle \leqslant \langle x, x \rangle \quad \forall x \in \mathcal{H}, \nonumber
\end{equation}
or equivalently
\begin{equation}
 \sum_i V^\dagger_i V_i \leqslant 1 \label{traceno}
\end{equation}
is called a quantum operation. The index set of the summation in \eqref{phi} and \eqref{traceno} can be in principle uncountable, which means that the ancilla Hilbert space used for the 
derivation of the quantum operation is non-separable, see p. $318$ in Ref. \onlinecite{Kraus}. The summation in \eqref{phi} is convergent with respect to the trace norm $\|\cdot\|_1$ and 
$ \sum_i V^\dagger_i V_i$ is convergent with respect to the ultraweak topology on $\mathcal{B}(\mathcal{H})$ 
generated by the seminorms
\begin{equation}
 \mathcal{B}(\mathcal{H}) \rightarrow \mathbb{R}_+, \quad X \rightarrow \big| \mathrm{Tr}\{XA\} \big| \quad \text{with}\quad A \in \mathcal{B}_1(\mathcal{H}). \nonumber
\end{equation}
The proof can be found in Lemma $2.1$ of Ref. \onlinecite{Kraus}. If $\sum_i V^\dagger_i V_i=1$ then $\phi$ is trace preserving. A quantum operation is a 
completely positive map due to the involvement of partial trace over the ancilla Hilbert space in its derivation \cite{Hellwig}.

Any $X \in \mathcal{B}_1(\mathcal{H})$ can be written as 
\begin{equation}
 X/\|X\|_1= \rho_1-\rho_2+i \left(\rho_3-\rho_4\right) \nonumber
\end{equation}
with four density operators $\rho_i\in \mathcal{D}(\mathcal{H})$ ($i \in {1,2,3,4}$). Therefore, $\phi$ can be extended to a linear mapping of $\mathcal{B}_1(\mathcal{H})$ into itself and the extension,
denoted by $\phi$ again, is defined by
\begin{equation}
 \phi(X)=\sum_i V_i X V^\dagger_i, \quad  \forall X \in \mathcal{B}_1(\mathcal{H}), \nonumber 
\end{equation}
where the summation is convergent with respect to the trace norm $\|\cdot\|_1$, see Lemma $2.1$ in \onlinecite{Kraus}.
 
Every bounded linear functional on $\mathcal{B}_1(\mathcal{H})$ is of the form $X \rightarrow \mathrm{Tr}\{XA\}$ with $A \in \mathcal{B}(\mathcal{H})$. This establishes a one-to-one correspondence
preserving linearity and the norm between $\mathcal{B}(\mathcal{H})$ and the dual of $\mathcal{B}_1(\mathcal{H})$, see Theorem $2$ on p. $47$ in Ref. \onlinecite{Schatten}. Thus, $\mathcal{B}(\mathcal{H})$ 
is isometrically isomorphic to the dual of 
$\mathcal{B}_1(\mathcal{H})$. Hence, the adjoint map $\phi^*$ to $\phi$ of $\mathcal{B}_1(\mathcal{H})$ is defined for arbitrary $A \in \mathcal{B}(\mathcal{H})$ by
\begin{equation}
 \mathrm{Tr}\{\phi(X)A\}=\mathrm{Tr}\{X\phi^*(A)\}, \nonumber
\end{equation}
where
\begin{equation}
\phi^*(A)=\sum_i V^\dagger_i A V_i. \label{phics} 
\end{equation}
According to Lemma $3.1$ in \onlinecite{Kraus}, $\phi^*$ is also a completely positive map and the summation in \eqref{phics} is convergent with respect to the ultraweak topology on $\mathcal{B}(\mathcal{H})$.

The aim of this paper is to investigate the behavior of the iterates $\phi^n$ and $\left(\phi^*\right)^n$ with the help of mean ergodic theorems.  
Classical mean ergodic theorems \cite{Krengel} are revisited for this purpose. As both $\phi$ and $\phi^*$ generate abelian discrete semigroups one may 
use the splitting theorem of Jacobs-DeLeeuw-Glicksberg \cite{Ja1, Ja2, DG1, DG2} by assuming that the semigroups have relatively weakly compact orbits, which is vital to the theory, see Eberlein in 
Ref. \onlinecite{Eber}. Here, we do not assume this property and we focus only on the mean ergodic property of $\phi$ and $\phi^*$, which even it is fulfilled does not guarantee that these
operators generate relatively weakly compact semigroups, see Example $8.27$ in \onlinecite{EFHN}. Another important direction with a considerable literature is the investigation of 
Schwarz maps or semigroups on  $C^*$- and von Neumann algebras \cite{Nagel} and this is closely related to quantum operations, because completely positive maps are Schwarz maps \cite{Groh1, Groh2}. This topic has a 
considerable literature and mean ergodicity can be obtained either via the Gelfand-Naimark-Segal construction \cite{Batkai}, where the semigroup on the algebra is extended to a semigroup of contractions on a 
Hilbert space, or in terms of the ergodic and quasi-ergodic projections \cite{Luczak}. These investigations consider not only discrete semigroups but also one-parameter semigroups and 
an overview of these results and further references can be found in Ref. \onlinecite{Emel}. The topic of convergence rates in mean ergodic theorems concerning both discrete and one-parameter semigroups 
is also an under active investigations \cite{Butzer, Gomilko1, Gomilko2}.

Asymptotic behavior of quantum operations has found application in many physical models. Fixed points of quantum operations have already been studied in one-dimensional quantum spin chains \cite{Bratteli} or in the 
context of quantum measurements \cite{Arias, Nagy}. Quantum operations may have a decoherence free subspace or in the spirit Jacobs-DeLeeuw-Glicksberg decomposition a reversible part of quantum states and this 
has also been investigated for open quantum walks \cite{Carbone}. In finite dimension, where mean ergodicity is immediate, properties of quantum operations have been discussed in detail \cite{Jaroslav, Wolff}.

In this paper, we shall be concerned with the mean ergodicity of $\phi$ and $\phi^*$ and we show 
that these completely positive maps can also be defined on the Hilbert space of Hilbert-Schmidt operators. This combined with some classical mean ergodic theorems results
that these maps on the space of Hilbert-Schmidt operators are mean ergodic. Then, some known results are revisited in order to characterize the fixpoints $\phi$ and $\phi^*$ and 
two supporting examples are provided. Since mean ergodicity of $\phi$ and $\phi^*$ is not guaranteed in general, therefore we revisit a classical result of Yosida and Kakutani in \onlinecite{Kakutani} without the 
assumption of weakly compactness of the power bounded operator. This is then applied to $\phi$ and $\phi^*$ in order to define the projectors of the subspace of reversible vectors and 
we recall also the uniform ergodic theory in order to clarify the circumstances, when the iterates $\phi^n$ and $\left(\phi^*\right)^n$ vanish on the complement of the reversible part with respect to 
the closed subspace, on which Ces\`{a}ro averages of these maps are defined. The special case of $\phi$ and $\phi^*$, when these maps are defined on the space of Hilbert-Schmidt operators, 
has the most potential to be applied in different physical models and it is pointed out to be a simple consequence of the results of Yosida and Kakutani or of the splitting theorem of 
Jacobs-DeLeeuw-Glicksberg. Here, we also provide examples of quantum operations
on both finite and infinite dimensional Hilbert spaces. 

This article is organized as follows. Fix points and mean ergodicity of $\phi$ and $\phi^*$ are the subjects of Sec. \ref{II}. It will be shown in Sec. \ref{III} that this approach
plays an important role in the asymptotic properties of the iterates for large $n$. As our problem is formulated on Banach spaces of linear operators on general separable Hilbert spaces 
we need to use several deep results from the theory of operators.
 
\section{Mean ergodic theorems}
\label{II}

Let us fix some notations. An operator $T: \mathcal{B} \rightarrow \mathcal{B}$ is called power bounded on a Banach space $\mathcal{B}$
if the norms of the powers $T^n$ ($n\geqslant0$) are uniformly bounded.
\begin{proposition}
\label{prop1}
The two completely positive maps $\phi: \mathcal{B}_1(\mathcal{H}) \rightarrow \mathcal{B}_1(\mathcal{H})$ and $\phi^*: \mathcal{B}(\mathcal{H}) \rightarrow \mathcal{B}(\mathcal{H})$ 
have operator norm at most $1$ and are thus power bounded.
\end{proposition}
\begin{proof}
By Lemma $3.1$ in Ref. \onlinecite{Paulsen} the matrix
 \begin{equation}
  \begin{pmatrix}
     I & X \\
     X^\dagger & X^\dagger X  
     \end{pmatrix}=\begin{pmatrix}
     I & 0 \\
     X^\dagger & 0  \end{pmatrix} \begin{pmatrix}
     I & X \\
     0 & 0  
     \end{pmatrix} \geqslant 0, \quad \forall X \in \mathcal{B}(\mathcal{H}) \nonumber
 \end{equation}
with $I$ being the identity operator on $\mathcal{H}$ and by applying the positive map $\phi^{(2)}$ we have
\begin{equation}
\phi^{*}(X^\dagger)\phi^{*}(X)\leqslant  \|\phi^{*}(I)\|\phi^{*}(X^\dagger X),  \label{Kadison}
\end{equation}
where we have used $\|\phi^{*}(I)\| I \geqslant \phi^{*}(I)$, because $\phi^{*}(I)$ is a positive operator. From \eqref{Kadison} follows that $\|\phi^{*}(U)\| \leqslant \|\phi^{*}(I)\|$ if $U$ is unitary.
From Corollary $1$ of Ref. \onlinecite{Russo} we know that the operator norm $\|\cdot\|_{\mathcal{B}(\mathcal{H})}$ of $\phi^*$ is equal to $\sup_U \|\phi^{*}(U)\|$ where supremum is over the unitaries
in $\mathcal{B}(\mathcal{H})$. Consequently
\begin{equation}
 \|\phi^* \|_{\mathcal{B}(\mathcal{H})}=\|\phi^{*}(I)\|=\|\sum_i V^\dagger_i V_i \| \leqslant 1. \nonumber
\end{equation}
For a more general argumentation see Theorem $1.3.3$ in Ref. \onlinecite{Stormer}. Theorem $1.10.12$ of \onlinecite{Megginson} says that the bounded linear map $\phi$ and its dual $\phi^*$ have the
same norm $\|\phi^* \|_{\mathcal{B}(\mathcal{H})}=\|\phi \|_{\mathcal{B}_1(\mathcal{H})}$, where $\|\cdot\|_{\mathcal{B}_1(\mathcal{H})}$ is the operator norm  of $\phi$, and by thus $\phi$ is also a contraction.
\end{proof}

In the subsequent part we discuss very general mean ergodic theorems for a power bounded operator $T: \mathcal{B} \rightarrow \mathcal{B}$ on a Banach space $\mathcal{B}$. 
We are interested in the convergence of the averages $A_n(T)=n^{-1} \sum^n_{i=1} T^i$,
the so-called Ces\`{a}ro mean of the first $n$ iterates of $T$. We introduce the following two closed linear subspaces \cite{Krengel}:
\begin{equation}
 \mathcal{B}_{\text{me}}(T)=\{x \in \mathcal{B}: \lim A_n(T)x \,\, \text{exists}\}, \nonumber
\end{equation}
and the fixed space of $T$
\begin{equation}
\text{F}(T)=\{x \in \mathcal{B}: T x=x\}=\text{Ker}(I-T). \nonumber
\end{equation}
Let $\mathcal{B}'$ be the space of all continuous functionals on $\mathcal{B}$, which is also a Banach space and called the dual space of $\mathcal{B}$. 
The adjoint operator $T^* : \mathcal{B}' \rightarrow \mathcal{B}' $ of $T$ is defined by
$f(Tx)=(T^*f)(x)$ for $f \in \mathcal{B}'$ and $x \in \mathcal{B}$. Then we have the following splitting theorem \cite{Yosida}
\begin{theorem}
 \label{theorem1}
 Let $T$ be a power bounded linear operator on a Banach space $\mathcal{B}$. Then 
 \begin{equation}
  \mathcal{B}_{\text{me}}(T)=\text{F}(T) \oplus \overline{\text{Rng}(I-T)}. \nonumber
 \end{equation}
The linear operator $\mathcal{P}x=\lim A_n(T) x$ assigned to $x \in \mathcal{B}_{\text{me}}(T)$ is the projection on $\mathcal{B}_{\text{me}}(T)$
onto $\text{F}(T)$ along $\overline{\text{Rng}(I-T)}$. We have $\mathcal{P}=\mathcal{P}^2=T\mathcal{P}=\mathcal{P}T$, and for every $y\in \mathcal{B}$ the assertions
\begin{enumerate}[label=\alph*)]
 \item $\lim A_n(T) y=0$
 \item $y \in \overline{\text{Rng}(I-T)}$
 \item $f(y)=0$ for all $f \in \{g \in \mathcal{B}' : T^*g=g\}$
\end{enumerate}
are equivalent.
\end{theorem}
\begin{proof}
 The theorem is mostly due to K. Yosida \cite{Yosida} and proved in \onlinecite{Krengel}, see Theorem $1.3$ on p. $73$-$74$.  
\end{proof}
If $\mathcal{B}_{\text{me}}(T)=\mathcal{B}$ then $T$ is called mean ergodic. A Banach space $\mathcal{B}$ is reflexive, when the canonical embedding map from $\mathcal{B}$ into 
the double dual $\mathcal{B}''$ is surjective. Power bounded linear operators on a reflexive Banach space are mean ergodic \cite{Lorch}. The approach of K. Yosida is based on weakly compact linear operators, 
which map the unit ball of the Banach space on a relatively weakly compact set \cite{Yosida}. Every bounded linear operator defined on a reflexive Banach space is weakly compact. One is able to construct a 
non-reflexive Banach space such that all contractions 
on the space are mean ergodic \cite{Fonf}, however this topic is outside the scope of this paper. As both $\mathcal{B}_1(\mathcal{H})$ and $\mathcal{B}(\mathcal{H})$ are not 
reflexive in general, the maps $\phi$ and $\phi^*$ are not always mean ergodic. 

The set of Hilbert-Schmidt operators
\begin{equation}
\mathcal{B}_2(\mathcal{H})=\{X \in \mathcal{B}(\mathcal{H}): \|X\|_2=\mathrm{Tr}\{ X^\dagger X\} <\infty \} \nonumber
\end{equation}
is a Banach space and with the scalar product $\langle X, Y\rangle_{\text{HS}}=\mathrm{Tr}\{ X^\dagger Y\}$ becomes a Hilbert space, which means that it is a reflexive space. In fact, 
$\mathcal{B}_2(\mathcal{H})$ is a two sided $*$-closed ideal in $\mathcal{B}(\mathcal{H})$ and
\begin{equation}
\mathcal{B}_1(\mathcal{H})\subseteq\mathcal{B}_2(\mathcal{H})\subseteq\mathcal{B}(\mathcal{H}). \label{tart} 
\end{equation}

\begin{corollary}
Let $\phi^*_2$ be the restriction of $\phi^*$ to $\mathcal{B}_2(\mathcal{H})$.  If $\|\sum_i V_i V^\dagger_i \| \leqslant 1$ then $\phi^*_2$ is mean ergodic. \label{coru}
\end{corollary}
\begin{proof}
By Proposition \ref{prop1} $\phi^*$ is a contraction and $\phi^{*}(X^\dagger)\phi^{*}(X)\leqslant \phi^{*}(X^\dagger X)$ for all $X \in \mathcal{B}(\mathcal{H})$. 
Since $\phi^{*}(X^\dagger)=\phi^{*}(X)^\dagger$, we have for every $X \in \mathcal{B}_2(\mathcal{H})$
\begin{equation}
 \|\phi^*(X)\|^2_2=\mathrm{Tr}\left\{\phi^{*}(X)^\dagger \phi^{*}(X)\right\}\leqslant \mathrm{Tr}\left\{\phi^{*}(X^\dagger X)\right\}=\mathrm{Tr}\left\{X^\dagger X \sum_i V_i V^\dagger_i \right\}
\end{equation}
and $X^\dagger X \in \mathcal{B}_1(\mathcal{H})$. Thus
\begin{equation}
\mathrm{Tr}\left\{X^\dagger X \sum_i V_i V^\dagger_i \right\} \leqslant \|X^\dagger X\|_1 \|\sum_i V_i V^\dagger_i \| \leqslant \mathrm{Tr}\left\{X^\dagger X\right\}=\|X\|^2_2.
\end{equation}
Hence, $\phi^*_2= \phi^* \big|_{\mathcal{B}_2(\mathcal{H})}$ is also a contraction and thus power bounded. The Banach space $\mathcal{B}_2(\mathcal{H})$ is a Hilbert space, hence it is reflexive. 
The theorem of Lorch \cite{Lorch} implies the mean ergodicity of $\phi^*_2$.
\end{proof}

\begin{remark}
It is interesting to note that the adjoint of $\phi^*_2$ is
\begin{equation}
 \phi_2(X)=\sum_i V_i X V^\dagger_i,  \nonumber
\end{equation}
where we have used the properties of the inner product $\langle X, Y\rangle_{\text{HS}}$. Thus, the map in \eqref{phi} can also be extended
to the space of Hilbert-Schmidt operators. 
\end{remark}

\begin{corollary}
\label{cor2}
If $\|\sum_i V_i V^\dagger_i \| \leqslant 1$ then the maps $\phi^*_2$ and $\phi_2$ have the property $\text{F}(\phi^*_2)=\text{F}(\phi_2)$.
\end{corollary}
\begin{proof}
We have shown in Corollary \ref{coru} that $\phi^*_2$ is a contraction. Since $\|\phi^*_2\|_{\mathcal{B}_2(\mathcal{H})}=\|\phi_2\|_{\mathcal{B}_2(\mathcal{H})}$ both $\phi^*_2$ and $\phi_2$ 
are contractions on $\mathcal{B}_2(\mathcal{H})$. 
Then the proof follows from 
Theorem $8.6$ and Corollary $8.7$ in Ref. \onlinecite{EFHN} and here
we give a sketch of the arguments. Take $X \in \text{F}(\phi^*_2)$ which implies
\begin{equation}
 \langle X, \phi_2(X)\rangle_{\text{HS}}=\langle \phi^*_2(X), X\rangle_{\text{HS}}=\|X\|^2_2. \nonumber
\end{equation}
Since $\phi_2$ is a contraction, it follows that
\begin{eqnarray}
 \|\phi_2(X)-X\|^2_2&=&\|\phi_2(X)\|^2_2-2\text{Re} \langle X, \phi_2(X)\rangle_{\text{HS}}+\|X\|^2_2 \nonumber \\
 &=&\|\phi_2(X)\|^2_2-\|X\|^2_2 \leqslant 0. \nonumber
\end{eqnarray}
Thus, $\phi_2(X)=X$ or $X \in \text{F}(\phi_2)$. Consequently, $\text{F}(\phi^*_2) \subseteq \text{F}(\phi_2)$. The opposite inclusion is proven in the same fashion, but 
starting with $X \in \text{F}(\phi_2)$.
\end{proof}

Now let us give examples of fixed spaces.

\begin{example}
Let $\mathcal{H}=\mathbb{C}^2$ with its standard basis $e_0=\begin{pmatrix}
    1\\
    0
  \end{pmatrix} $ and $e_1=\begin{pmatrix}
    0\\
    1
  \end{pmatrix} $. We consider the following two Pauli matrices
\begin{equation}
\sigma_x=\begin{pmatrix} 0&1\\1&0 \end{pmatrix}, \quad \sigma_y=\begin{pmatrix} 0&-i\\i&0 \end{pmatrix}, \nonumber
\end{equation}
and the quantum operation
\begin{eqnarray}
&&\phi(X)=V_1 X V^\dagger_1 + V_2 X V^\dagger_2, \nonumber \\
&&V_1=\sqrt{p} \sigma_x, \quad V_2=\sqrt{1-p} \sigma_y, \nonumber
\end{eqnarray}
where $p\in [0,1]$ and $X \in \mathcal{M}_{2}(\mathbb{C})$, i.e., the set of all $2 \times 2$ matrices over $\mathbb{C}$. Note that $\phi^*=\phi$.

If $p=0$ then
\begin{equation}
 \text{F}(\phi)=\left\{\begin{pmatrix} a&-b\\b&a \end{pmatrix}, \quad  \forall a,b \in \mathbb{C} \right\}. \nonumber
\end{equation}

If $p \in (0,1)$ then
\begin{equation}
 \text{F}(\phi)=\left\{\begin{pmatrix} a&0\\0&a \end{pmatrix}, \quad  \forall a \in \mathbb{C} \right\}. \nonumber
\end{equation}

If $p=1$ then
\begin{equation}
 \text{F}(\phi)=\left\{\begin{pmatrix} a&b\\b&a \end{pmatrix}, \quad  \forall a,b \in \mathbb{C} \right\}. \nonumber
\end{equation}
\end{example}

\begin{example} \label{ex1}
Let $\mathcal{H}=\ell^2$, the space of square-summable sequences. The left and right shift operators are defined by
\begin{eqnarray}
&&S_L : (a_1, a_2, a_3, a_4 \dots) \mapsto (a_2, a_3, a_4, a_5 \dots) , \nonumber \\
&&S_R : (a_1, a_2, a_3, a_4 \dots) \mapsto (0, a_1, a_2, a_3, \dots), \nonumber
\end{eqnarray}
and using these contractions we construct the following quantum operation
\begin{eqnarray}
&&\phi(X)=V_1 X V^\dagger_1 + V_2 X V^\dagger_2, \nonumber \\
&&V_1=\sqrt{p} S_L, \quad V_2=\sqrt{1-p} S_R, \quad p\in (0,1), \nonumber
\end{eqnarray}
where $X \in \mathcal{B}_1(\ell^2)$. It is immediate that $\phi$ is a contraction and
\begin{equation}
 \left \langle x, \sum_{i=1,2}V^\dagger_iV_i x \right \rangle= \left \langle x, \left(p S^\dagger_L S_L+ (1-p) S^\dagger_R S_R\right) x \right \rangle < \langle x, x \rangle, \quad \forall x \in \ell^2, \nonumber
\end{equation}
or equivalently
\begin{equation}
 \sum_{i=1,2}V^\dagger_iV_i= p S^\dagger_L S_L+ (1-p) S^\dagger_R S_R < I \nonumber
\end{equation}
with $I$ being the identity operator on $\ell^2$. We consider first $X$ to be
\begin{equation}
 X : (a_1, a_2, a_3, a_4 \dots) \mapsto (x_1 a_1, x_2 a_2, x_3 a_3, \dots) \quad \text{with} \quad \sum^\infty_{i=1} |x_i| < \infty\nonumber
\end{equation}
and $p$ has a fixed value. Then $\phi(X)=X$ has the following solution
\begin{equation}
 x_i=x_1 \frac{f^{(i)}(p)}{p^{i-1}}, \quad i>1, \nonumber
\end{equation}
where
\begin{eqnarray}
 f^{(i)}(p)&=&\sum^{i-1}_{j=0} (-1)^j p^j a^{(i)}_j, \nonumber \\
 a^{(i)}_j&=&\begin{cases}
            1,& j=0,\\
            a^{(i-1)}_j+a^{(i-1)}_{j-1},& 0< j <i-1,\\
            1, & j=i-1 \quad \text{and} \quad i \text{ is odd}, \\
            0, & j=i-1 \quad \text{and} \quad i \text{ is even}.
           \end{cases} \nonumber
\end{eqnarray}
We find that $f^{(i)}(p)/p^{i-1} > 1 $ and one obtains by induction the following relation 
\begin{equation}
 \frac{f^{(i+1)}(p)}{p^{i}}-\frac{f^{(i)}(p)}{p^{i-1}}=\frac{(1-p)^i}{p^i}>0, \quad i>1. \nonumber
\end{equation}
It is immediate that the solution to $\phi(X)=X$ is not a trace class operator, and even more so is not bounded. Extending the presented approach to more general
$X$ operators, 
\begin{equation}
 X : (a_1, a_2, a_3, a_4 \dots) \mapsto \left(\sum^\infty_{i=1} x_{1i} a_i, \sum^\infty_{i=1} x_{2i} a_i, \sum^\infty_{i=1} x_{3i} a_i, \dots\right ) \nonumber
\end{equation}
we find $\phi(X)=X$ has the following solution
\begin{equation}
x_{ij}=x_{1j} \frac{f^{(i)}(p)}{p^{i-1}}, \quad i>1 \quad \text{and} \quad j\geqslant 1.  \nonumber 
\end{equation}
Therefore,
\begin{equation}
\text{F}(\phi)=\{0\}, \nonumber 
\end{equation}
here $0$ denotes the null operator.
\end{example}
\begin{remark}
Both averages $A_n(\phi)$ and $A_n(\phi^*)$ can be evaluated for elements of $\mathcal{D}(\mathcal{H})$ due to 
$\mathcal{D}(\mathcal{H}) \subset \mathcal{B}_1(\mathcal{H}) \subseteq \mathcal{B}(\mathcal{H})$. However, there might be cases when
\begin{eqnarray}
\mathcal{D}(\mathcal{H}) \cap \text{F}(\phi) = \emptyset,  \nonumber \\
\mathcal{D}(\mathcal{H}) \cap \text{F}(\phi^*) = \emptyset.  \nonumber
\end{eqnarray}
If $\sum_i V^\dagger_i V_i=1$ then 
\begin{equation}
 \{X \in \mathcal{B}_1(\mathcal{H}): A_iX=XA_i \quad \text{for all} \,i \}\subseteq \text{F}(\phi) \nonumber
\end{equation}
and
\begin{equation}
 \{X \in \mathcal{B}(\mathcal{H}): A_iX=XA_i \quad \text{for all} \, i \}\subseteq \text{F}(\phi^*). \nonumber
\end{equation}
For more details, see Ref. \onlinecite{Arias}.
\end{remark}

We shall now discuss a particular property of completely positive maps. As the following argumentation is the same for both $\phi$ and  $\phi^*$, we consider the case of $\phi^*$. 
The following result is known and we formulate it in the context discussed here in this manuscript. 

\begin{lemma}
Let $p_i$ be strictly positive numbers with $\sum_i p_i=1$ and let $\phi^*_i$ be commuting completely positive maps. If $\phi^*=\sum_i p_i \phi^*_i$ then
\begin{equation}
 \text{F}(\phi^*)= \bigcap_i \text{F}(\phi^*_i). \nonumber
\end{equation}
\end{lemma}
\begin{proof}
The proof is mostly due to A. Brunel \cite{Brunel} and  M. Falkowitz \cite{Falkowitz}. An elegant way of proving it is to use that the identity is an extreme point in the convex set consisting of all
contractions, an application of the Krein-Milman Theorem. For further details, see Lemma $1.14$ on p. $82$ in \onlinecite{Krengel}.
\end{proof}

Mean ergodic theorems enable us to understand the convergence of the averages $A_n(T)$ and the fixed space of $T$. In the subsequent section we show 
how these averages are used in evaluations of iterates $\phi^n$ and $\left(\phi^*\right)^n$.

\section{Main result}
\label{III}

In this chapter we are concerned with the iterates $\phi^n$ and $\left(\phi^*\right)^n$. In order to characterize them we need
to slightly modify Theorem $2$ of K. Yosida and S. Kakutani in \onlinecite{Kakutani}.
\begin{theorem}
\label{theorem2}
Let $T$ be a power bounded linear operator on a Banach space $\mathcal{B}$ and $\lambda$ any complex number with $|\lambda|=1$. Then $T_\lambda=T/\lambda$ is a 
power bounded operator such that 
\begin{enumerate}[label=\alph*)]
 \item $A_n(T_\lambda)$ converges strongly to $\mathcal{P}_\lambda$, the projection of $\mathcal{B}_{\text{me}}(T_\lambda)$ onto $\text{F}(T_\lambda)$ along  $\overline{\text{Rng}(I-T_\lambda)}$,
 \item $\mathcal{P}_\lambda=\mathcal{P}^2_\lambda=T_\lambda \mathcal{P}_\lambda=\mathcal{P}_\lambda T_\lambda$,
 \item $\lambda \neq \mu$ implies $\text{F}(T_\lambda)\cap \text{F}(T_\mu)=\{0\}$,
 \item $\mathcal{P}_\lambda \neq 0$ iff $\lambda$ is an eigenvalue of $T$.
\end{enumerate}
\end{theorem}
\begin{proof}
It is immediate from the relation
\begin{equation}
 \sup_k \|T^k \|_{\text{op}}=\sup_k \|T^k_\lambda \|_{\text{op}} \nonumber
\end{equation}
that $T_\lambda$ is power bounded, where $\|.\|_{\text{op}}$ is the operator norm. The strong convergence
\begin{equation}
 \lim_{n \to \infty} \| A_n(T_\lambda) x - \mathcal{P}_\lambda x \|=0, \quad \forall x \in \mathcal{B}_{\text{me}}(T_\lambda) \nonumber
\end{equation}
with the norm $\|.\|$ of $\mathcal{B}$ and the relations 
\begin{equation}
 \mathcal{P}_\lambda=\mathcal{P}^2_\lambda=T_\lambda \mathcal{P}_\lambda=\mathcal{P}_\lambda T_\lambda \nonumber
\end{equation}
are clear from Theorem \ref{theorem1}. Furthermore, we also have
\begin{equation}
 \mathcal{P}_\lambda T =T \mathcal{P}_\lambda=\lambda \mathcal{P}_\lambda. \label{use1}
\end{equation}
We have $F(T_\lambda)=\text{Ker}(\lambda I-T)$ and $F(T_\mu)=\text{Ker}(\mu I-T)$, so every $x$ in
the intersection of these spaces satisfies $\lambda x = T x = \mu x$. If $\lambda \neq \mu$ then $x=0$.
Hence, the statement in $c)$ is proven. Since $\mathcal{P}_\lambda$ is projecting onto $F(T_\lambda)=\text{Ker}(\lambda I-T)$
and $\text{Ker}(\lambda I-T) \neq \{0\}$ when $\lambda$ is an eigenvalue, we arrive at the statement in $d)$.
\end{proof}

We introduce the peripheral point spectrum
\begin{equation}
 \sigma_{p,p}(T)=\left\{\lambda \in \mathbb{C} : \text{Ker}(\lambda I-T) \neq \{0\}, \, |\lambda|=1\right\} \label{sigma}
\end{equation}
and
\begin{equation}
\mathcal{B}(T)=\bigcap_{\lambda \in \sigma_{p,p}(T)} \mathcal{B}_{\text{me}}(T_\lambda). \label{set}
\end{equation}
The peripheral point spectrum plays an essential role in the splitting theorem of Jacobs-DeLeeuw-Glicksberg \cite{EFHN} and its properties are under active investigations, see for example 
Refs. \onlinecite{Gluck1, Gluck2}. 

After all these preparations we return to the iterates of $\phi$ and $\phi^*$. The next result is essentially a consequence of Proposition \ref{prop1} and Theorem \ref{theorem2}.
\begin{corollary} \label{coruse}
Let $\{P_\lambda\}_{\lambda \in \sigma_{p,p}(\phi)}$ be the set of projectors associated according to Theorem \ref{theorem2} with $\phi: \mathcal{B}_1(\mathcal{H}) \rightarrow \mathcal{B}_1(\mathcal{H})$ 
and similarly $\{Q_\lambda\}_{\lambda \in \sigma_{p,p}(\phi^*)}$ with $\phi^*: \mathcal{B}(\mathcal{H}) \rightarrow \mathcal{B}(\mathcal{H})$. Set
\begin{eqnarray}
 S_\phi (X)&=&\phi(X)-\sum_{\lambda \in \sigma_{p,p}(\phi)} \lambda P_\lambda(X), \quad \forall X \in \mathcal{B}(\phi) \subseteq \mathcal{B}_1(\mathcal{H}), \nonumber \\
 S_{\phi^*}(X)&=&\phi^*(X)-\sum_{\lambda \in \sigma_{p,p}(\phi^*)} \lambda Q_\lambda(X), \quad \forall X \in \mathcal{B}(\phi^*) \subseteq \mathcal{B}(\mathcal{H}), \nonumber
\end{eqnarray}
then the iterates of $\phi$ and $\phi^*$ are given by the formulas
\begin{eqnarray}
\phi^n(X)&=& \sum_{\lambda \in \sigma_{p,p}(\phi)} \lambda^n P_\lambda(X) + S^n_\phi(X), \quad \forall X \in \mathcal{B}(\phi), \nonumber \\
\left(\phi^*\right)^n(X)&=& \sum_{\lambda \in \sigma_{p,p}(\phi^*)} \lambda^n Q_\lambda(X) + S^n_{\phi^*}(X), \quad \forall X \in \mathcal{B}(\phi^*) \nonumber
\end{eqnarray}
for $n\geqslant 1$. $\lambda$ is an eigenvalue of $S_\phi$ ($S_{\phi^*}$) with eigenvector in $\mathcal{B}(\phi)$ ($\mathcal{B}(\phi^*)$) iff is an eigenvalue of $\phi$ ($\phi^*$) and 
$\lambda \not \in \sigma_{p,p}(\phi)$ ($\lambda \not \in \sigma_{p,p}(\phi^*)$).
\end{corollary}
\begin{proof}
We consider only the case of $\phi$, because the proof for $\phi^*$ is exactly along the same lines. Let $X_\lambda \in \text{Ker}(\lambda I-\phi)$ with $\lambda \in \sigma_{p,p}(\phi)$. Then we have
$\lambda X_\lambda \in \text{Ker}(\lambda I-\phi)$, $\|\lambda X_\lambda\|_1=\|X_\lambda\|_1$, and $\text{Ker}(\lambda I-\phi) \subset \mathcal{B}_{\text{me}}(\phi)$, because
\begin{eqnarray}
 &&\lim_{n \to \infty} \|A_n(\phi)(X_\lambda)\|_1=\lim_{n \to \infty} \frac{1}{n} \left| \frac{1-\lambda^n}{1-\lambda}\right| \|X_\lambda\|_1=0, \quad \lambda \neq 1, \nonumber \\
 &&\lim_{n \to \infty} \|A_n(\phi)(X_\lambda)\|_1=\|X_\lambda\|_1, \quad \text{when} \quad \lambda=1 \in \sigma_{p,p}(\phi). \nonumber
\end{eqnarray}
Since $\text{Ker}(\lambda I-\phi) \cap \text{Ker}(\mu I-\phi)=\{0\}$ for $\lambda \neq \mu$, we have
\begin{equation}
 \sum_{\lambda \in \sigma_{p,p}(\phi)} \lambda P_\lambda(X_\mu) \in \text{Ker}(\mu I-\phi), \quad \forall \mu \in \sigma_{p,p}(\phi) \nonumber
\end{equation}
and
\begin{equation}
 \mathrm{Rng} \left(\sum_{\lambda \in \sigma_{p,p}(\phi)} \lambda P_\lambda \right)=\overline{\mathrm{lin}} 
\left\{X \in \mathcal{B}_1(\mathcal{H}): \, \phi(X)=\lambda X, \, \lambda \in \sigma_{p,p}(\phi) \right\}. \nonumber
\end{equation}
By Theorem \ref{theorem2} we have for all $X \in \mathcal{B}(\phi)$
\begin{equation}
 P_\lambda S_\phi (X) = P_\lambda \left(\phi(X)-\sum_{\lambda \in \sigma_{p,p}(\phi)} \lambda P_\lambda(X) \right)=\lambda P_\lambda(X) - \lambda P_\lambda(X)=0 \nonumber 
\end{equation}
and similarly $S_\phi P_\lambda(X)=0$ for all $\lambda \in \sigma_{p,p}(\phi)$. It is immediate also that
\begin{equation}
 \phi S_\phi (X)=S_\phi \phi(X)=S^2_\phi (X). \nonumber
\end{equation}
Then, for all $X \in \mathcal{B}(\phi)$
\begin{eqnarray}
 \phi^2(X)&=&\left(\sum_{\lambda \in \sigma_{p,p}(\phi)} \lambda P_\lambda+ S_\phi \right)^2 (X) \nonumber \\
 &=& \left(\sum_{\lambda \in \sigma_{p,p}(\phi)} \lambda P_\lambda\right)^2 (X)+S^2_\phi(X), \nonumber
\end{eqnarray}
but according to Theorem \ref{theorem2} $P_\lambda P_\mu (X)=0$ whenever $\lambda \neq \mu$, which yields
\begin{equation}
\phi^2(X)=\sum_{\lambda \in \sigma_{p,p}(\phi)} \lambda^2 P_\lambda(X)+S^2_\phi(X). \nonumber 
\end{equation}
By induction we obtain the required result for all powers of $n$. To prove the last statement, let $\lambda \neq 0$ be such an eigenvalue of $\phi$ with eigenvector $X_\lambda \neq 0$ that 
\begin{equation}
 \phi (X_\lambda)=\lambda X_\lambda, \quad X_\lambda \in \mathcal{B}(\phi) \nonumber
\end{equation}
and $\lambda \not \in \sigma_{p,p}(\phi)$. Then, for all $\mu \in \sigma_{p,p}(\phi)$ 
\begin{equation}
 \mu P_\mu(X_\lambda)= P_\mu \phi (X_\lambda)=\lambda P_\mu (X_\lambda) \nonumber
\end{equation}
and since $\mu \neq \lambda$, we obtain $ P_\mu (X_\lambda)=0$. Therefore, 
\begin{equation}
 \lambda X_\lambda=\phi (X_\lambda)=\sum_{\mu \in \sigma_{p,p}(\phi)} \mu P_\mu(X_\lambda)+S_\phi (X_\lambda)=S_\phi (X_\lambda). \nonumber
\end{equation}
Hence $\lambda$ is an eigenvalue of $S_\phi$. Conversely, let $\lambda \neq 0$ be an eigenvalue of $S_\phi$  with eigenvector $X_\lambda \neq 0$ such that 
\begin{equation}
S_\phi (X_\lambda)=\lambda X_\lambda, \quad X_\lambda \in \mathcal{B}(\phi). \nonumber 
\end{equation}
We have already shown that $\phi S_\phi(X)=S^2_\phi (X)$ and using this relation we obtain
\begin{equation}
 \phi (X_\lambda)=\frac{1}{\lambda}\phi S_\phi (X_\lambda)=\frac{1}{\lambda} S^2_\phi (X_\lambda) = \lambda X_\lambda, \nonumber
\end{equation}
$\lambda$ is also an eigenvalue of $\phi$. In order to show that $\lambda \not \in \sigma_{p,p}(\phi)$, we assume that there exists a $\mu \in \sigma_{p,p}(\phi)$ such that $\lambda = \mu$. Since
\begin{equation}
 A_n(\phi/\mu)(X_\lambda)=\frac{1}{n} \left(\frac{\phi}{\mu}+\frac{\phi^2}{\mu^2}+ \dots + \frac{\phi^n}{\mu^n} \right)(X_\lambda) =X_\lambda, \quad n \geqslant 1 \nonumber
\end{equation}
and the limit $n \to \infty$  with the help of Theorem \ref{theorem1} results $P_\mu(X_\lambda)=X_\lambda$. It is a contradiction because
\begin{equation}
 P_\mu(X_\lambda)=\frac{1}{\lambda}P_\mu S_\phi (X_\lambda)=0, \nonumber 
\end{equation}
where we have used the already proved relation $P_\mu S_\phi =0$. Thus, $\lambda \not \in \sigma_{p,p}(\phi)$.
\end{proof}

\begin{remark}
 It is important to note that the sets $\mathcal{B}(\phi)$ and $\mathcal{B}(\phi^*)$ play an essential role in the above result. Therefore, it is natural to ask under which 
 conditions $\mathcal{B}(\phi)=\mathcal{B}_1(\mathcal{H})$ and $\mathcal{B}(\phi^*)=\mathcal{B}(\mathcal{H})$. This is the case, when $\phi$ and $\phi^*$ are weakly compact operators \cite{Yosida} or
 $\mathcal{H}$ is finite-dimensional, which imply that $\mathcal{B}_1(\mathcal{H})$ and $\mathcal{B}(\mathcal{H})$ are reflexive Banach spaces \cite{Lorch}.
 In general, the mean ergodicity of both $\phi$ and $\phi^*$ may be not simultaneously fulfilled.
\end{remark}

\begin{corollary}
If $\|\sum_i V_i V^\dagger_i \| \leqslant 1$ then the iterates of $\phi^*_2: \mathcal{B}_2(\mathcal{H}) \rightarrow \mathcal{B}_2(\mathcal{H})$ and its adjoint map $\phi_2$ 
are given by the formulas
\begin{eqnarray}
\left(\phi^*_2\right)^n(X)&=& \sum_{\lambda \in \sigma_{p,p}(\phi^*_2)} \lambda^n Q_{\lambda,2}(X) + S^n_{\phi^*_2}(X), \quad \forall X \in \mathcal{B}_2(\mathcal{H}) \nonumber \\
\phi^n_2(X)&=& \sum_{\lambda \in \sigma_{p,p}(\phi_2)} \lambda^n P_{\lambda,2}(X) + S^n_{\phi_2}(X), \quad \forall X \in \mathcal{B}_2(\mathcal{H}). \nonumber 
\end{eqnarray}
On the closures of the linear spans 
\begin{equation}
\overline{\mathrm{lin}} \left\{X: \, \phi^*_2(X)=\lambda X, \, \lambda \in \sigma_{p,p}(\phi^*_2) \right\} \,\text{and} \quad \overline{\mathrm{lin}} 
\left\{X: \, \phi_2(X)=\lambda X, \, \lambda \in \sigma_{p,p}(\phi_2) \right\} \nonumber
\end{equation}
$\phi^*_2$ and respectively $\phi_2$ restrict to unitary operators.
\end{corollary}
\begin{proof}
The first part is immediate from Theorem \ref{theorem2} and Corollary \ref{coruse}. One has to use the following facts:
$\mathcal{B}_2(\mathcal{H})$ is a Hilbert space; $\text{F}(T_\lambda)=\text{Ker}(\lambda I-T)$; $\phi^*_2$ and $\phi_2$ are mean ergodic and contractions due to Corollary \ref{coru}. 
Furthermore, $\phi_2/\lambda$ with $\lambda \in \sigma_{p,p}(\phi_2)$ is also a contraction and by applying Corollary \ref{cor2} one obtains $\text{F}(\phi^*_2/\lambda^*)=\text{F}(\phi_2/\lambda)$ and therefore
when $\lambda \in \sigma_{p,p}(\phi_2)$ implies $\lambda^* \in \sigma_{p,p}(\phi^*_2)$. We have also
\begin{equation}
 \overline{\mathrm{lin}} \left\{X: \, \phi^*_2(X)=\lambda X, \, \lambda \in \sigma_{p,p}(\phi^*_2) \right\}=\overline{\mathrm{lin}} 
\left\{X: \, \phi_2(X)=\lambda X, \, \lambda \in \sigma_{p,p}(\phi_2) \right\} \nonumber
\end{equation}
on which $\phi_2$ restricts to a unitary operator and $\phi^*_2$ is its inverse.
\end{proof}

This result is similar to the findings of B. Sz\H{o}kefalvi-Nagy and C. Foia\c{s}, a contraction on a Hilbert space
defines a decomposition of the Hilbert space into two parts, where on one of them the contraction acts as a
unitary operator \cite{SZNF}. On the other hand this can also be viewed as an application of the splitting theorem of Jacobs-DeLeeuw-Glicksberg, see Example $16.25$ in Ref. \onlinecite{EFHN}.

In order to investigate the limit $n \to \infty$ of the formulas obtained in Corollary \ref{coruse} one has to determine the spectrum of $\phi$ or $\phi^*$. It is of interest that 
$\|S_\phi\|_{\mathcal{B}_1(\mathcal{H})} <1$ and $\|S_{\phi^*}\|_{\mathcal{B}(\mathcal{H})} <1$, because then the asymptotic space can be identified through the projectors 
$\{P_\lambda\}_{\lambda \in \sigma_{p,p}(\phi)}$ and 
$\{Q_\lambda\}_{\lambda \in \sigma_{p,p}(\phi^*)}$. An obvious choice is that $\phi^*$ and $\phi$ are compact operators. The image of the unit ball under a compact operator is relatively compact, because 
the spectrum of a compact operator contains the cluster point $\{0\}$ and only eigenvalues, which form the point spectrum. According to Theorem \ref{theorem2} an eigenvalue $\lambda$ with 
$|\lambda|=1$ is not an eigenvalue of either $S_\phi$ or $S_{\phi^*}$. As both $\phi^*$ and $\phi$ are contractions, their spectrum is contained in the closed unit disc, we have
\begin{equation}
 \lim_{n \to \infty} S^n_\phi=\lim_{n \to \infty} S^n_{\phi^*}=0. \nonumber
\end{equation}
It turns out that one can make an even more general statement. But first we have to introduce the concept of quasi-compact operators. An operator $T$ on a Banach space $\mathcal{B}$ is 
quasi-compact if there exists an integer $n$ and a compact operator $V$ with $\|T^n-V \|_{\mathcal{B}}<1$.
This leads to the uniform ergodic theory of  K. Yosida and S. Kakutani, which is the consequence of applying Theorem \ref{theorem2} to quasi-compact operators \cite{Kakutani}. Thus, we are able
to prove the following result:

\begin{corollary}
Let $\phi$ and $\phi^*$ be quasi-compact operators. Then, there exists constants $\epsilon_\phi, \epsilon_{\phi^*}>0$ and $M_\phi, M_{\phi^*}>0$ such that
\begin{equation}
\|S^n_\phi\|_{\mathcal{B}_1(\mathcal{H})} \leqslant \frac{M_\phi}{\left(1+\epsilon_\phi\right)^n} \quad \text{and} \quad 
\|S^n_{\phi^*}\|_{\mathcal{B}(\mathcal{H})} \leqslant \frac{M_{\phi^*}}{\left(1+\epsilon_{\phi^*}\right)^n}. \nonumber
\end{equation}
\end{corollary}
\begin{proof}
As the proof can be found in \onlinecite{Kakutani} (see Lemma $3.3$ and Theorem $4$) we show only the cornerstones of the argumentation. Quasi-compactness of $\phi$ and $\phi^*$ implies that 
$\lambda \phi$ and $ \lambda \phi^*$ are quasi-compact and mean-ergodic for each $\lambda$ on the unit circle. Furthermore, the ranges of the projectors $\{P_\lambda\}_{\lambda \in \sigma_{p,p}(\phi)}$ and 
$\{Q_\lambda\}_{\lambda \in \sigma_{p,p}(\phi^*)}$ are finite dimensional. Therefore, both $S_\phi$ and $S_{\phi^*}$ are quasi-compact and the unit circle belongs to their resolvent set, 
which is equivalent to the existence of $\epsilon_\phi, \epsilon_{\phi^*}>0$ and $M_\phi, M_{\phi^*}>0$, which fulfill the relations of the statement.
\end{proof}

\begin{remark}
It is worth to point out that the mean ergodicity of a power bounded bounded in Theorems \ref{theorem1} and \ref{theorem2} can be guaranteed by considering weakly quasi-compactness, a weaker assumption 
then weak compactness, see Theorem $3$ in Ref. \onlinecite{Kakutani}. The question, under which conditions is a completely positive map compact or quasi-compact, is left
open for now. Although the uniform ergodic theory is an old result, there are recent developments with respect to the iterates of quasi-compact operators, see for example Ref. \onlinecite{Nagler}.
\end{remark}

\begin{example}
Let us reconsider the quantum operation in Example \ref{ex1}. $\mathcal{M}_{2}(\mathbb{C})$ with the scalar product $\langle X, Y \rangle_{\text{HS}} =\mathrm{Tr} \{X^\dagger Y\}$ is a Hilbert space. The
orthonormal basis is chosen to be
\begin{eqnarray}
 &X_{1}=\begin{pmatrix} 1&0\\0&1\end{pmatrix}/\sqrt{2}, \quad X_{2}=\begin{pmatrix} -1&0\\0&1\end{pmatrix}/\sqrt{2}, \nonumber \\
 &X_{3}=\begin{pmatrix}  0&1\\1&0\end{pmatrix}/\sqrt{2}, \quad X_{4}=\begin{pmatrix} 0&-1\\1&0\end{pmatrix}/\sqrt{2}. \nonumber
\end{eqnarray}
Then
\begin{equation}
 \phi(X)=\lambda_1 P_{\lambda_1}(X)+\lambda_2 P_{\lambda_2}(X)+S_{\phi}(X), \nonumber
\end{equation}
where
\begin{eqnarray}
 \lambda_1=1 \quad \text{with} \quad P_{\lambda_1}(X)=\mathrm{Tr}\{X^\dagger_{1} X\} X_{1}, \nonumber \\
 \lambda_2=-1 \quad \text{with} \quad P_{\lambda_2}(X)=\mathrm{Tr}\{X^\dagger_{2} X\} X_{2}, \nonumber \\
 S_{\phi}(X)=(2p-1) \mathrm{Tr}\{X^\dagger_{3} X\} X_{3}+(1-2p) \mathrm{Tr}\{X^\dagger_{4} X\} X_{4}. \nonumber
\end{eqnarray}
It is immediate that
\begin{equation}
 \phi^n(X)=P_{\lambda_1}(X)+(-1)^n P_{\lambda_2}(X)+(2p-1)^n \mathrm{Tr}\{X^\dagger_{3} X\} X_{3}+(1-2p)^n \mathrm{Tr}\{X^\dagger_{4} X\} X_{4} \nonumber
\end{equation}
and
\begin{equation}
 \|S^n_\phi\|_{\mathcal{M}_{2}(\mathbb{C})}\leqslant |1-2p|^n, \nonumber
\end{equation}
where $\|.\|_{\mathcal{M}_{2}(\mathbb{C})}$ is the operator norm on $\mathcal{M}_{2}(\mathbb{C})$. When $p \in (0,1)$ we have the relations $|1-2p| <1$ and
$S^n_\phi \rightarrow 0$ as $n \to \infty$.
\end{example}

\begin{example}
Let $\mathcal{H}$ be a separable Hilbert space with orthonormal basis $e_n$ ($n=0,1,2, \dots$). We define the operators $a$ and $a^\dagger$ by
\begin{equation}
 ae_n=\sqrt{n} e_{n-1}, \quad a^\dagger e_n=\sqrt{n+1} e_{n+1}. \nonumber
\end{equation}
In physics, $\mathcal{H}$ is the symmetric Fock space, in which $e_n$ is associated with the number state of $n$ bosons, and $a^\dagger$ and $a$ are interpreted as 
creation and annihilation operators of bosons.

For the sake of simplicity we set $|n\rangle=e_n$ and then a typical vector $x$ has the unique expansion $x=\sum^\infty_{i=0} x_n |n\rangle$. The orthonormal basis $e'_n$ ($n=0,1,2, \dots$) of the dual of 
$\mathcal{H}$ has the property $e'_n(e_m)=0$ when $n \neq m$ and $e'_n(e_n)=1$. We define $\langle n |=e'_n$. The operator $a^\dagger$ is the adjoint of $a$, because
\begin{equation}
 \langle a^\dagger x, y \rangle = \sum^\infty_{n=0} \sqrt{n+1 }x^*_n y_{n+1}=\langle x, a y \rangle. \nonumber  
\end{equation}
Furthermore, they are unbounded operators and have the same domain of definition, which is dense in $\mathcal{H}$. We consider the 
following quantum operation
\begin{eqnarray}
&&\phi(X)=V_1 X V^\dagger_1 + V_2 X V^\dagger_2, \nonumber \\
&&V_1=\sqrt{p}I, \quad V_2= \sqrt{1-p} e^{-i \pi a^\dagger a}, \quad p \in (0,1), \nonumber
\end{eqnarray}
where $X \in \mathcal{B}_1(\mathcal{H})$ and $I$ is the identity map. Although $a^\dagger$ and $a$ are unbounded, $V_2$ is a bounded operator with spectral radius one. 

$\phi$ can be extend to $\mathcal{B}_2(\mathcal{H})$. Now, with the help of the Hilbert-Schmidt 
scalar product we define the following orthogonal projections
\begin{equation}
 P_{n,m}(X)=\mathrm{Tr}\left\{| m \rangle \langle n | X\right\} | n \rangle \langle m | \quad X \in \mathcal{B}_2(\mathcal{H}) \nonumber
\end{equation}
with $n,m=0,1,2,\dots$. Then
\begin{eqnarray}
 \phi^n(X)&=&\sum^\infty_{n=0} P_{n,n}(X)+\sum^{\infty}_{k=1}\sum^\infty_{n=0}\Big[ P_{n+2k,n}(X)+ P_{n,n+2k}(X)\Big]  \nonumber \\
 &+&(2p-1) \sum^{\infty}_{k=0}\sum^\infty_{n=0}\Big[ P_{n+2k+1,n}(X)+ P_{n,n+2k+1}(X)\Big]. \nonumber
\end{eqnarray}
Thus,
\begin{equation}
 S_{\phi}(X)=(2p-1) \sum^{\infty}_{k=0}\sum^\infty_{n=0}\Big[ P_{n+2k+1,n}(X)+ P_{n,n+2k+1}(X)\Big] \nonumber
\end{equation}
and
\begin{equation}
 \|S^n_\phi\|_{\mathcal{B}_2(\mathcal{H})}\leqslant |1-2p|^n, \nonumber
\end{equation}
where $\|.\|_{\mathcal{B}_2(\mathcal{H})}$ is the operator norm on $\mathcal{B}_2(\mathcal{H})$.
\end{example}


\bibliographystyle{amsplain}

\begin{thebibliography}{99}

\bibitem{Hellwig} K.-E. Hellwig and K. Kraus, Commun. Math. Phys. {\bf 11}, 214 (1969).

\bibitem{Kraus} K. Kraus, Ann. Phys. {\bf 64}, 311 (1971).

\bibitem{GKS} V. Gorini, A. Kossakowski and E. C. G. Sudarshan, J. Math. Phys. {\bf 17}, 821 (1976).

\bibitem{Lindblad} G. Lindblad, Comm. Math. Phys. {\bf 48}, 119 (1976).

\bibitem{Spohn}  H. Spohn, Rev. Mod. Phys. {\bf 52}, 569 (1980).

\bibitem{Nielsen} M. A. Nielsen and I. L. Chuang, {\it Quantum Computation and
Quantum Information} (Cambridge University Press, Cambridge, UK, 2000).

\bibitem{Wiseman} H. M. Wiseman and G. J. Milburn, {\it Quantum Measurement and Control} (Cambridge University
Press, Cambridge, UK, 2010).

\bibitem{Joos} E. Joos, H. D. Zeh, C. Kiefer, D. Giulini, J. Kupsch, and I.-O. Stamatescu
{\it Decoherence and the Appearance of a Classical World in Quantum 
Theory} (Springer-Verlag, Berlin, 1996).

\bibitem{Krengel} U. Krengel, {\it Ergodic Theorems} (de Gruyter, Berlin, 1985).

\bibitem{comment} It is not necessary that the Hilbert space is separable, because properties of quantum operations are not related to the separability criteria, see pg. 318 in Ref. \cite{Kraus}. Furthermore,
relations between different operator spaces are true for all Hilbert spaces \cite{Schatten}. However, there are two reasons to not consider nonseparable Hilbert spaces. First, quantum mechanics and its applications
are built around separable Hilbert spaces which is due to the history of the theory. Second, Gleason's Theorem in J. Math. and Mech. {\bf 6}, 885 (1957), which charachterizes all states on a separable Hilbert space, 
cannot be generalized to nonseparable Hilbert spaces without assuming the continuum hypothesis, see Eilers-Horst Theorem in Int. J. Theor. Phys. {\bf 13}, 419 (1975). 
The consequences of this interesting statement are beyond the scope of the theory. 

\bibitem{Schatten} R. Schatten, {\it Norm Ideals of Completely Continuous Operators} (Springer-Verlag, Berlin, 1970).

\bibitem{Ja1} K. Jacobs, Math. Z. {\bf 64}, 298 (1956).

\bibitem{Ja2} K. Jacobs, Math. Z. {\bf 67}, 83 (1957).

\bibitem{DG1} K. DeLeeuw and I. Glicksberg, Amer. Math. Soc. {\bf 65}, 134 (1959).

\bibitem{DG2} K. DeLeeuw and I. Glicksberg, Acta Math. {\bf 105}, 63 (1961).

\bibitem{Eber} W. F. Eberlein, Trans. Amer. Math. Soc. {\bf 67}, 217 (1949).

\bibitem{EFHN} T. Eisner, B. Farkas, M. Haase, and R. Nagel, {\it Operator theoretic aspects of ergodic theory} (Springer-Verlag, Cham, 2015).

\bibitem{Nagel} R. Nagel, editor, {\it One-parameter semigroups of positive operators}, vol. {\bf 1184} (Cham, Springer, 1986).

\bibitem{Groh1} U. Groh, J. Operator Theory {\bf 10}, 31 (1983).

\bibitem{Groh2} U. Groh, Math. Jpn. {\bf 29}, 395 (1984).

\bibitem{Batkai} A. B\'atkai, U. Groh, D. Kunszenti-Kov\'acs, and M. Schreiber, Semigroup Forum {\bf 84}, 8 (2012).

\bibitem{Luczak} K. Kielanowicz and A. Luczak, Publ. Mat., Barc. {\bf 64}, 283 (2020).

\bibitem{Emel} Eduard Yu. Emel'yanov, {\it Non-spectral asymptotic analysis of one-parameter operator semigroups}, vol. {\bf 173}  (Birkh\"auser, Basel, 2007).

\bibitem{Butzer} P. Butzer and U. Westphal, Indiana Univ. Math. J. {\bf 20}, 1163 (1971).

\bibitem{Gomilko1} A. Gomilko, M. Haase, and Y. Tomilov, Math. Res. Lett. {\bf 18}, 201 (2011).

\bibitem{Gomilko2} A. Gomilko, M. Haase, and Y. Tomilov, J. Anal. Math. {\bf 118}, 545 (2012).

\bibitem{Bratteli} O. Bratteli, P. J\o{}rgensen, A. Kishimoto and R. Werner, J. Operator Theory {\bf 43}, 97 (2000).

\bibitem{Arias} A. Arias, A. Gheondea, and S. Gudder, J. Math. Phys. {\bf 43}, 5872 (2002).

\bibitem{Nagy} G. Nagy, J. Math. Phys. {\bf 49}, 022110 (2008).

\bibitem{Carbone} R. Carbone and A. Jen\u{c}ov\'a, Ann. Henri Poincar\'e {\bf 21}, 155 (2020).

\bibitem{Jaroslav} J. Novotn\'{y}, G. Alber, and I. Jex, J. Phys. A: Math. Theor. {\bf 45}, 485301 (2012).

\bibitem{Wolff} M. M. Wolf: Quantum channels and operations--guided tour, Online Lecture Notes (2012).

\bibitem{Kakutani} K. Yosida and S. Kakutani, Ann. Math. {\bf 42}, 188 (1941).

\bibitem{Paulsen} V. Paulsen, {\it Completely Bounded Maps and Operator Algebras} (Cambridge University Press, Cambridge, UK, 2003).

\bibitem{Russo} B. Russo and H. A. Dye, Duke Math. J. {\bf 33}, 413 (1966).

\bibitem{Stormer} E. St\o{}rmer, {\it Positive linear maps of operator algebras} (Springer-Verlag, Berlin, 2013).

\bibitem{Megginson} R. E. Megginson, {\it An Introduction to Banach Space Theory} (Springer-Verlag, New York, 1998).

\bibitem{Yosida} K. Yosida, Proc. Imp. Acad. Tokyo {\bf 14}, 292 (1938).

\bibitem{Lorch} E. R. Lorch, Bull. Amer. Math. Soc. {\bf 45}, 945 (1939).

\bibitem{Fonf} V. P. Fonf, M. Lin, and P. Wojtaszczyk, Israel J. Math. {\bf 179}, 479 (2010).

\bibitem{Brunel} A. Brunel, Lecture notes, Springer-Verlag, t. 160, 7-17 (1970).

\bibitem{Falkowitz} M. Falkowitz, PAMS {\bf 38}, 553 (1973).

\bibitem{Gluck1} J. Gl\"uck, Positivity {\bf 20}, 307 (2016).

\bibitem{Gluck2} J. Gl\"uck, J. Math. Anal. Appl. {\bf 453}, 317 (2017).

\bibitem{SZNF} B. Sz-Nagy and C. Foia\c{s}, Acta Sci. Math. Szeged {\bf 21}, 251 (1960).

\bibitem{Nagler} J. Nagler, J. Math. Anal. Appl. {\bf 462}, 347 (2018).

\end{thebibliography}

\end{document}